\documentclass[12pt]{article}

\usepackage{amsmath}
\usepackage{amsfonts}
\usepackage{amsthm}
\usepackage{enumerate}
\usepackage{graphicx}
\usepackage{subcaption}
\usepackage{color}
\usepackage{cite}

\newtheorem{lemma}{Lemma}

\newtheorem{theorem}{Theorem}
\theoremstyle{definition}

\title{Geometric Spanning Cycles in Bichromatic Point Sets}

\author{B. Joeris \thanks{University of Waterloo, Waterloo, Canada. Email:
bjoeris@uwateloo.ca}
\and I. Urrutia \thanks{University of Waterloo, Waterloo, Canada. Email:
ihurruti@uwaterloo.ca} \and J. Urrutia \thanks{Instituto de Matem{\'a}ticas, UNAM,
Mexico.
Email:
urrutia@matem.unam.mx. Research supported by project number 178379 Conacyt, M{\'e}xico.}}

\date{}

\begin{document}
%\linenumbers
\maketitle

\begin{abstract}
Given a set of points in the plane each colored either red or blue, we find
non-self-intersecting geometric spanning cycles of the red points and of the
blue points such that each edge of the red spanning cycle is crossed at most
three times by the blue spanning cycle and vice-versa.
\end{abstract}

\section{Introduction}
A \emph{geometric graph} is a graph embedded in the plane
with edges that are straight-line segments. A set of points is in general position if no three points of the set are collinear.
In this paper, a bichromatic point set is a finite set of points $S$ in general position, partitioned into two disjoint color classes $S_R$ and $S_B$ (red and blue.)

Several problems have been studied that involve finding geometric graphs on sets
of red and blue points. 
Alternating paths in bichromatic point sets in convex position were studied in
\cite{akiyama1990simple}. Alternating paths in general position were
studied in \cite{abellanas1999bipartite}.
Alternating paths in points with more
than two colors were studied in \cite{merino2006length}.
\cite{tokunaga1996intersection} examined non-self-intersecting geometric spanning trees of the red points and the blue
points and found a tight bound on the minimum number of intersection points
between the red and blue spanning trees.
\cite{Kano2005301} considered the case of more than two colors, and studied the
 number of intersections for monochromatic spanning trees and for
monochromatic spanning cycles.
\cite{merino2005intersection} obtained a tight bound on the number of
intersections in monochromatic perfect matchings.
\cite{kano2013discrete} looked at points and lines in the plane lattice, and
studied the number of crossings for alternating matchings and monochromatic
spanning trees.

In \cite{tokunaga1996intersection}, Tokunaga also showed that there exist
non-intersecting geometric spanning paths $P_R$ and $P_B$ of the red and blue points respectively, such that each
edge of $P_R$ is crossed at most once by $P_B$, and vice-versa. 

One may then wonder if a similar result is possible for cycles; i.e., is it possible to construct spanning cycles with ``few'' intersections on bichromatic point sets. In particular, we wonder for what values of $k$ does the
following statement hold: there exist non-intersecting geometric spanning
cycles $C_R$ and $C_B$ of the red and blue points respectively, such that each
edge of $C_R$ is crossed at most $k$ times by $C_B$, and vice-versa.
In \cite{tokunaga1996intersection}, Tokunaga conjectured that this statement is
true when $k=2$.
It is easy to see that the statement is false with $k=1$.
We show here that the statement is true when $k=3$.

\begin{theorem}
\label{thm:main}
Given any bichromatic point set in general position, there exists a non-self-intersecting geometric spanning cycle of the red points and a
non-self-intersecting geometric spanning cycle of the blue points such that
each edge of the red spanning cycle is crossed by at most three edges of the blue spanning cycle, and vice-versa.
\end{theorem}

\subsection{Definitions and Notation}
If $X$ is a set of points and $y$ is a point outside the convex hull of $X$, we say that $y$ \emph{sees} a point $x\in X$ (with respect to $X$), if the line segment $(x,y)$ intersects the convex hull of $X$ only at $x$.
In other words, if the convex hull of $X$ were opaque, then $y$ could see $x$.
In particular, $x$ must be a vertex of the convex hull in order for $y$ to see it.
If $X$ and $Y$ are two sets of points with disjoint convex hulls, and $x\in X$ and $y\in Y$, we say that $x$ and $y$ \emph{see each other} (with respect to $X$ and $Y$) if $y$ sees $x$ with respect to $X$ and $x$ sees $y$ with respect to $Y$.

Let $X$ be a set of points in the plane and $p\not\in X$.
The \emph{radial order} of $X$ about $p$ is the cyclic  list $(x_1, x_2, \dots, x_n)$ taking all the
elements $x\in X$ ordered clockwise around $p$.
The \emph{interval between $x_i$ and $x_j$ in the radial order} is the set consisting of $\{x_i,x_{i+1},\ldots,x_j\}$, if $i\leq j$, or $\{x_i,x_{i+1},\ldots,x_n,x_1,\ldots,x_j\}$ if $j<i$.
If $y$ is in the interval between $x_i$ and $x_j$ in the radial order, then we say that $y$ lies \emph{between $x_i$ and $x_j$ in the radial order}.

If $R$ and $B$ are disjoint sets of red and blue points, respectively, then a
\emph{blob} in the radial order of $R\cup B$ about $p$ is a maximal
monochromatic set of consecutive elements in the radial order.
Because the blobs are monochromatic, it is natural to speak of \emph{red blobs} and \emph{blue blobs}, containing red and blue points respectively.

Given a radial order $(x_1,\ldots,x_n)$ of $R\cup B$ with respect to a point $p\not\in R\cup B$, for any $i\in\{1,\ldots,n\}$, $x_{i-1}$ is called the point \emph{before} $x_i$ and $x_{i+1}$ is called the point \emph{after} $x_i$, where the indices are taken modulo $n$.
For any blob $X$ in this radial ordering, the \emph{first point of $X$} is the point $x_i\in X$ such that $x_{i-1}\not\in X$.
Similarly, the \emph{last point of $X$} is the point $x_i \in X$ such that $x_{i+1}\not\in X$.
We say that $Y$ is the \emph{previous red (resp.\ blue) blob before} $X$ if $Y$ is a blue blob and there is no red (resp.\ blue) point between the last point of $Y$ and the first point of $X$ in the radial order.
Similarly, we say that $Z$ is the \emph{next red (resp.\ blue) blob after} $X$ if there is no red (resp.\ blue) point between the the last point of $X$ and the first point of $Z$ in the radial order.

\subsection{Jump Configurations}
If $X$ is a blob, and $Y$ is the next blob after $X$ of the same color as $X$, then a \emph{jump edge} from $X$ to $Y$ is a line segment $(x,y)$ where $x\in X$, $y\in Y$ such that $x$ and $y$ see each other with respect to $X$ and $Y$ and the angle between $x$ and $y$ with respect to the point $p$ is at most $\pi$.
This last condition ensures that each point $z$ on the line segment $(x,y)$ would, if added to the radial order, lie between $x$ and $y$. See Figure \ref{fig:jump-edges}.

\begin{figure}[h]
\centering
\includegraphics[width=5cm]{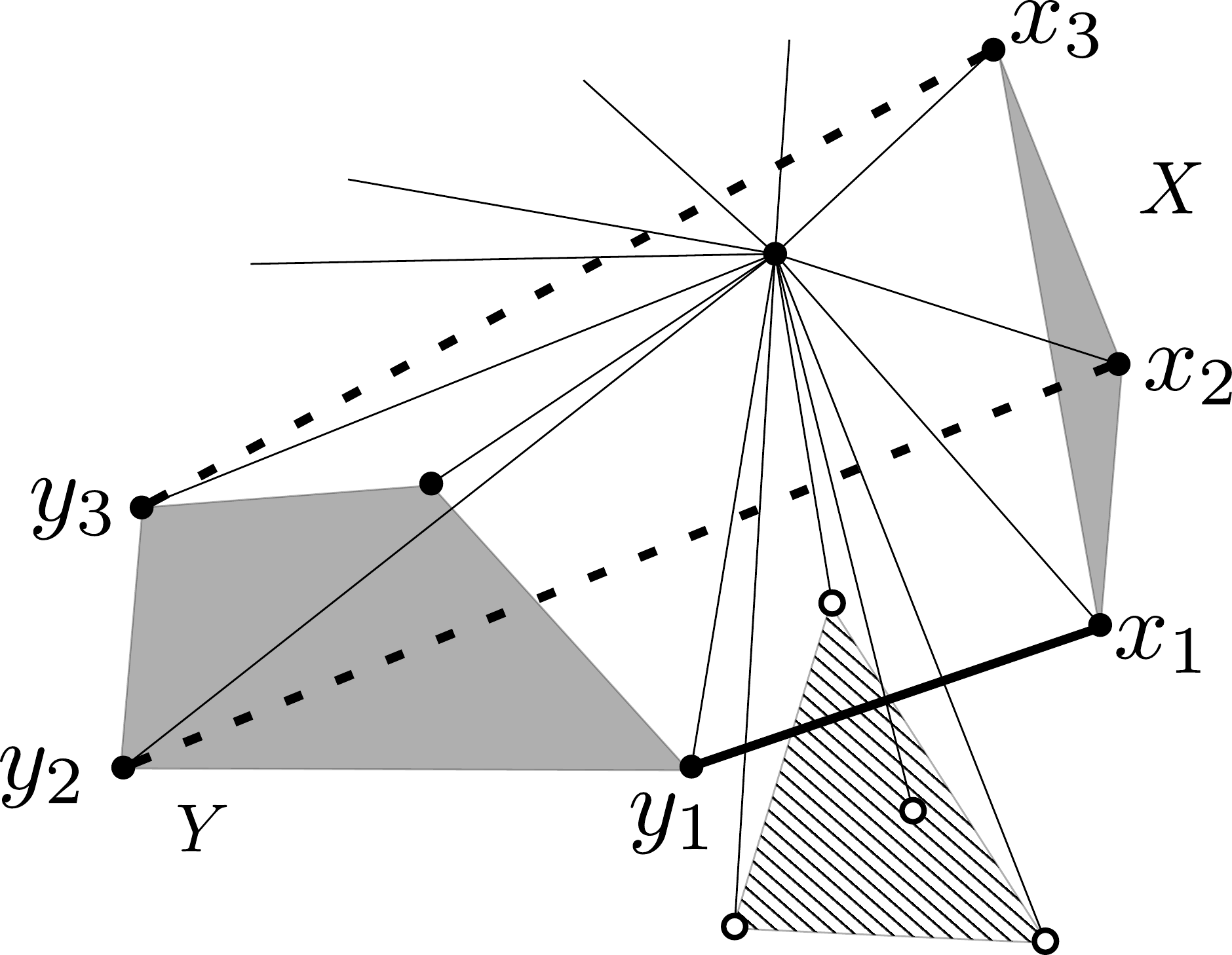}
\caption{Jump edges from a blob $X$ to the next blob of the same color, $Y$. $x_1$ to $y_1$ is a jump edge. $x_2$ to $y_2$ is not a
jump edge; it intersects the interior of the convex hulls of the blobs
containing $x_2$ and $y_2$.
$x_3$ to $y_3$ is not a jump edge because the angle from $x_3$ to $y_3$ is
greater than $\pi$.}
\label{fig:jump-edges}
\end{figure}

A \emph{jump configuration} $J$ is a collection of jump edges, one blue edge from each blue blob to the next  blue blob and one red edge from each red blob to the next red blob, such that the blue edges do not cross each other, the red edges do not cross each other, and, for each blob $X$, the two edges with endpoints in $X$ do not share a common endpoint unless $X$ contains only one point.

In this section we show that a jump configuration can be completed to a pair of non-intersecting spanning cycles of the two color classes by adding a spanning path within each blob, such that each edge of these cycles will be crossed at most three times, except where a specific structure, called a 4-forcing, appears.
4-forcings will be defined later in this section.

\begin{lemma}
  \label{lem:jump-config-blob-order}
  If there is a jump configuration $J$ in the radial order with respect to $p$, then for all blobs $X$, the angle from the first point in $X$ to the last point in $X$ is strictly less than $\pi$.
  Hence, for any $z$ in the convex hull of $X$, if $z$ were added to the radial order, it would lie in the interval from the first point in $X$ to the last point in $X$, which implies that the convex hulls of the blobs are disjoint.

  \begin{proof}
    Let $(b,a)$ be the jump edge in $J$ between the blob before $X$ and the blob after $X$.
    By the definition of a jump edge, the angle from $b$ to $a$ is strictly less than $\pi$.
    Also, $X$ is contained in the interval between $b$ and $a$ in the radial order, so the angle from the first point in $X$ to the last point in $X$ is strictly less than the angle from $b$ to $a$, which is less than $\pi$.
  \end{proof}
\end{lemma}

The following two lemmas describe how jump edges in a jump configuration can intersect, and how jump edges can intersect the convex hulls of blobs.

\begin{lemma}
  \label{lem:jump-config-crossing}
  If $X_0$ is a blob, and, for $i\in\{1,2,3,4\}$, $X_i$ is the next blob after $X_{i-1}$, then for any jump configuration $J$, the only jump edges in $J$ that can cross the jump edge from $X_1$ to $X_3$ are the jump edges from $X_0$ to $X_2$ and from $X_2$ to $X_4$.
  In particular, no jump edges of the same color cross, and each jump edge is crossed by at most two jump edges of the opposite color.

  \begin{proof}
    Let $(x,y)$ be the jump edge from $X_1$ to $X_3$ in the jump configuration $J$, and let $(x',y')$ be any jump edge in $J$ such that $(x,y)$ and $(x',y')$ meet at a point $z$.
    By construction, if added to the radial order, $z$ would lie between $x$ and $y$, and between $x'$ and $y'$.
    Therefore the interval $I$ between $x$ and $y$ and the interval $I'$ between $x'$ and $y'$ intersect.
    If $I'\supseteq I$, then $(x',y')$ must be a jump edge from $X_1$ to $X_3$, so $(x',y')=(x,y)$.
    Otherwise, either $x'\in I$ or $y'\in I$, so, without loss of generality, $x'$ lies between $x$ and $y$ in the radial order.
    Therefore $x'\in X_1 \cup X_2 \cup X_3$.
    If $x'\in X_1$ and $y'$ is in the previous blob of the same color, then $(x,y)$ and $(x',y')$ do not cross by construction, and similarly if $x'\in X_3$ and $y'$ is in the next blob of the same color.
    If $x'\in X_2$, then $(x',y')$ is either the jump edge from $X_0$ to $X_2$ or the jump edge from $X_2$ to $X_4$.
  \end{proof}
\end{lemma}

\begin{lemma}
  \label{lem:jump-config-cross-blob}
  If $X_0$ is a blob, and, for $i\in\{1,2,3,4\}$, $X_i$ is the next blob after $X_{i-1}$, then for any jump configuration $J$, the only jump edges that can intersect the convex hull of $X_2$ are the jump edges from $X_0$ to $X_2$, from $X_1$ to $X_3$, and from $X_2$ to $X_4$.

  \begin{proof}
    Let $(y_1,y_2)$ be the jump edge between blobs $Y_1$ and $Y_2$ in the jump configuration $J$ such that $(y_1,y_2)$ intersects the convex hull of $X_1$ at a point $z$.
    Then $z$, if added to the radial order, would lie between the first point of $X_1$ and the last point of $X_1$.
    But $z$ lies between a point in $Y_1$ and a point in $Y_2$ in the radial order, so $X_1$ must be $Y_1$, $Y_2$, or the blob between $Y_1$ and $Y_2$ in the radial order.
  \end{proof}
\end{lemma}

When a jump edge passes through the convex hull of a blob $X$ (which can only happen when it is the jump edge from the previous blob before $X$ to the next blob after $X$), we want to find a spanning path within $X$ that crosses the jump edge as few times as possible. The following lemma gives a construction for such a path.

\begin{lemma}
  \label{lem:spanning-path-crossing}
  If $X$ is a finite set of points in general position, $x,y\in X$ are distinct points, and $\ell$ is a line, then there exists a non-self-intersecting spanning path $P$ of $X$ such that
  \begin{enumerate}[(i)]
  \item if $X$ lies entirely on one side of $\ell$, then $P$ does not cross $\ell$;
  \item if $x$ and $y$ lie on opposite sides of $\ell$, then $P$ crosses $\ell$ exactly once;
  \item if $x$ and $y$ lie on the same side of $\ell$, and $X$ contains points on both sides of $\ell$, then $P$ crosses $\ell$ exactly twice.
  \end{enumerate}
  
  \begin{proof}
    By induction on $n=|X|$.
    If $n=2$, then $X=\{x,y\}$, and so the one-edge path from $x$ to $y$ crosses $\ell$ zero times, if $x$ and $y$ are on the same side of $\ell$; or once if $x$ and $y$ are on opposite sides of $\ell$.

    If $n>2$, then $X\setminus\{x\}$ contains at least two points, and so $x$ sees at least two points of $X\setminus\{x\}$.
    One of these points, $x'$, is not $y$.
    If there are multiple choices for $x'$, choose $x'$ to lie on the same side of $\ell$ as $x$.
    By the induction hypothesis, there exists a spanning path $P'$ of $X\setminus\{v\}$ from $x'$ to $y$ satisfying (i) (ii) and (iii).
    $P'$ lies entirely within the convex hull of $X\setminus \{x\}$, and the edge from $x$ to $x'$ intersects the convex hull only at $x'$, so $P=P'\cup(x,x')$ is a non-self-intersecting spanning path of $X$ from $x$ to $y$.
    
    If $X$ lies entirely on one side of $\ell$, then $P$ lies entirely on one side of $\ell$ because it is contained in the convex hull of $X$, so (i) holds.

    Now suppose $x$ and $y$ lie on opposite sides of $\ell$.
    If $x'$ lies on the same side of $\ell$ as $x$ then $P$ crosses $\ell$ as many times as $P'$, which is once.
    If $x'$ lies on the opposite side of $\ell$ from $x$, then $X\setminus\{x\}$ cannot contain any point on the same side of $\ell$ as $x$, or else $x$ would see some point $x''$ on the same side of $\ell$, and $x''\neq y$ because $y$ is on the opposite side of $\ell$, so $x''$ would have been chosen over $x'$.
    Therefore when $x'$ lies on the opposite side of $\ell$ from $x$, $P'$ does not cross $\ell$, so $P$ crosses $\ell$ exactly once.
    
    Finally, suppose $x$ and $y$ lie on the same side of $\ell$, and that $X$ contains some point $z$ on the opposite side of $\ell$ from $x$ and $y$.
    If $x'$ lies on the same side of $\ell$ as $x$ and $y$, then $P$ crosses $\ell$ as many times as $P'$, which is twice.
    If $x'$ lies on the opposite side of $\ell$ from $x$ and $y$, then $P'$ crosses $\ell$ once, so $P$ crosses $\ell$ twice.
  \end{proof}
\end{lemma}

We now show how to construct a pair of monochromatic spanning cycles from a jump configuration by adding spanning paths within each blob.

The next lemma shows that there is exactly one case in which adding spanning paths can force our monochromatic spanning cycles to have edges that are crossed 4 times. We call this case a \emph{4-forcing} and define it as follows:
Suppose we are given a jump configuration $J$. Let $X_0,Y_1,X_1,Y_2,X_2$ be consecutive blobs such that the jump edge $(y_1,y_2)$ between $Y_1$ and $Y_2$ in $J$ intersects the convex hull of $X_1$ and crosses both the jump edge from $X_0$ to $X_1$ and the jump edge from $X_1$ to $X_2$, and such that the endpoints in $X_1$ of the jump edges from $X_0$ and to $X_2$ both lie on the same side of $(y_1,y_2)$.
Then this is called a \emph{4-forcing} in $J$, and $(y_1,y_2)$ is called the \emph{center edge} of the 4-forcing.
See Figure \ref{fig:4-forcing}.

\begin{figure}[h]
\centering
\includegraphics[width=6cm]{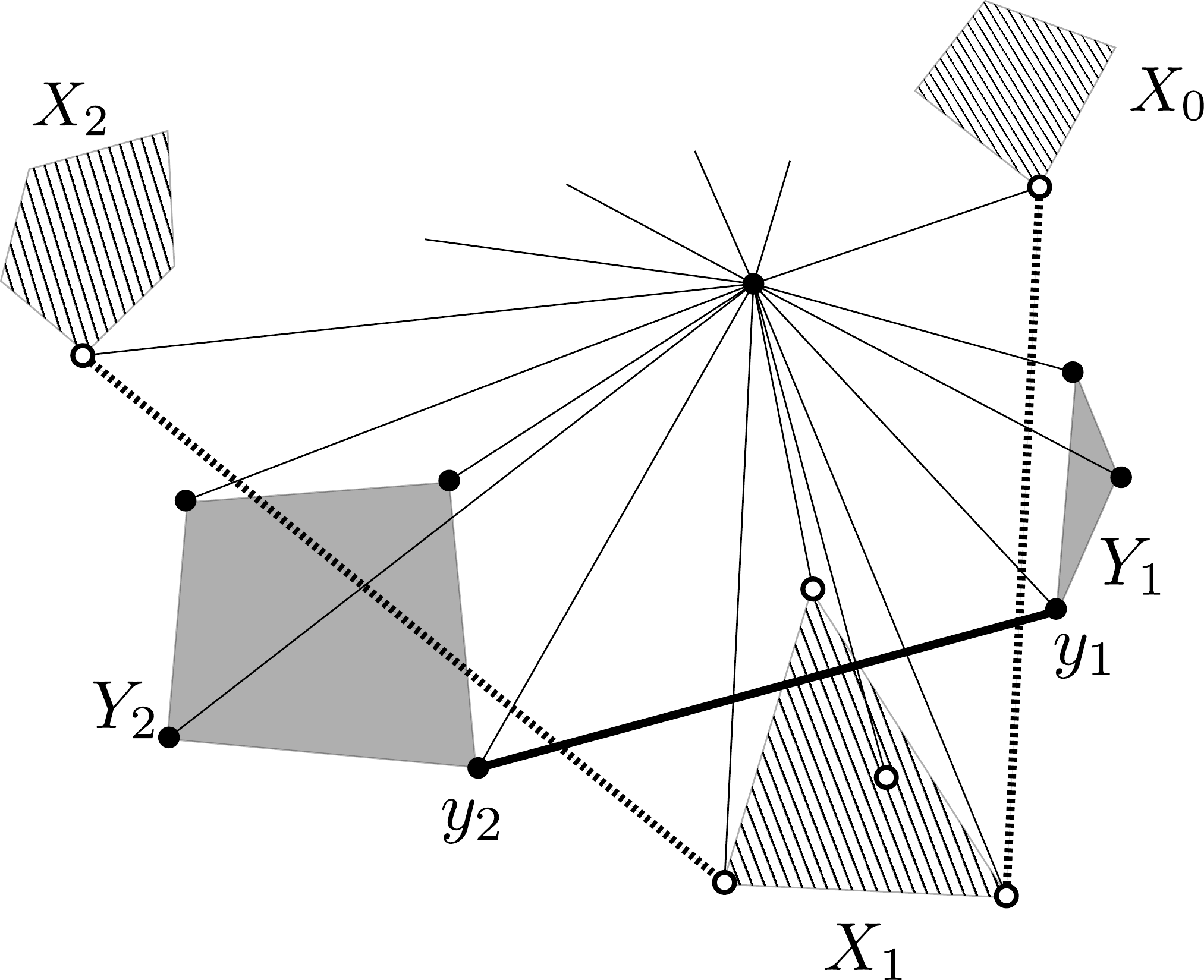}
\caption{4-forcing}
\label{fig:4-forcing}
\end{figure}

\begin{lemma}
  \label{lem:complete-jump-config}
  Given any jump configuration $J$,
  spanning paths of each blob can be added to $J$ to construct a pair of non-self-intersecting spanning cycles of the red and blue points respectively such that, if an edge $e$ is crossed more than three times by the opposite color cycle, then $e$ is in $J$ (and not in one of the spanning paths), and $e$ is the center edge of a 4-forcing.
  \begin{proof}
    % picture for notation
    For each blob $X$, let $a_X\in X$ be the point incident with the jump edge to $X$ from the previous blob of the same color, and let $b_X\in X$ be the point incident with the jump edge from $X$ to the next blob of the same color.
    If $W$ is the previous blob before $X$ and $Y$ is the next blob after $X$, let $\ell_X$ be the line through $b_W$ and $a_Y$, and note that $(b_W,a_Y)$ is the only jump edge that can cross through the convex hull of $X$.
    Let $P_X$ be a spanning path of $X$ from $a_X$ to $b_X$ minimizing the number of crossings of the $\ell_X$.

    By Lemma \ref{lem:jump-config-cross-blob}, no jump edge of the same color as $X$ can cross the convex hull of $X$, and hence no jump edge can cross $P_X$.
    Also, for any blob $X'\neq X$, the convex hulls of $X'$ and $X$ are disjoint, so $P_X$ and $P_{X'}$ cannot cross.
    By Lemma \ref{lem:jump-config-crossing}, two jump edges of the same color cannot cross, so the edges of  $J\cup \bigcup\limits_{\emph{all blobs $X$}}P_X$ form a pair of non-self-intersecting spanning-cycles.

    Let $e$ be an edge of one of these spanning cycles which is crossed at least 4 times by the other cycle.
    By Lemma \ref{lem:jump-config-cross-blob}, any edge of $P_X$ can only be crossed by the jump edge $(b_W,a_Y)$, and hence cannot be crossed 4 times, so $e$ is a jump edge.
    Without loss of generality, suppose $e$ is the jump edge between $W$ and $Y$.
    By Lemma \ref{lem:jump-config-crossing}, $e$ is crossed by at most two other jump edges (the two jump edges to and from $X$).
    The only blob for which $e$ could cross the convex hull is $X$, so the only blob whos spanning-path $e$ could cross is $P_X$.
    By Lemma \ref{lem:spanning-path-crossing}, $P_X$ crosses $e$ at most two times, with equality only if $a_X$ and $b_X$ lie on the same side of $\ell_X$.
    Therefore $e$ must be crossed by both of the jump edges incident with $X$, so $e$ is the center edge of a 4-forcing.
  \end{proof}
\end{lemma}

In the remainder of the paper, we focus on finding a good jump configuration such that when we add spanning cycles of each blob, as in the previous lemma, we have no 4-crossings \textemdash this will result in a pair of monochromatic spanning cycles in which each edge is crossed at most 3 times. In the following sections we show that it is possible to avoid 4-crossings by choosing $p$ carefully.

\section{Monster-Jumps}

If $B_1,R_1,B_2,R_2$ are consecutive blobs such that $B_1$ and $B_2$ are blue and $R_1$ and $R_2$ are red, then we say that $R_1$ to $R_2$ is a \emph{red monster-jump} if $|R_1|>1$ and the angle from the \emph{second} point in $R_1$ to the \emph{first} point in $R_2$ is at least $\pi$, and the line segment between the last point in $B_1$ and the first point in $B_2$ intersects the convex hull of $R_1$. See Figure \ref{fig:red-monster-jump}.

If $B_1,R_1,B_2,R_2$ are consecutive blobs such that $B_1$ and $B_2$ are blue and $R_1$ and $R_2$ are red, then we say that $B_1$ to $B_2$ is a \emph{blue monster-jump} if $|B_2|>1$ and the angle from the \emph{last} point in $B_1$ to the \emph{second to last} point in $B_2$ (i.e., the point before the last point in $B_2$) is at least $\pi$, and the line segment between the last point in $R_1$ and the first point in $R_2$ intersects the convex hull of $B_2$. See Figure \ref{fig:blue-monster-jump}.

Note the slight asymmetry between the definition of red and blue monster-jump.
In particular, if the colors are reversed and the plane reflected, then red monster-jumps become blue monster-jumps, and vice-versa.

\begin{figure}
\centering
\begin{subfigure}[b]{0.4\textwidth}
\centering
\includegraphics[width=\textwidth]{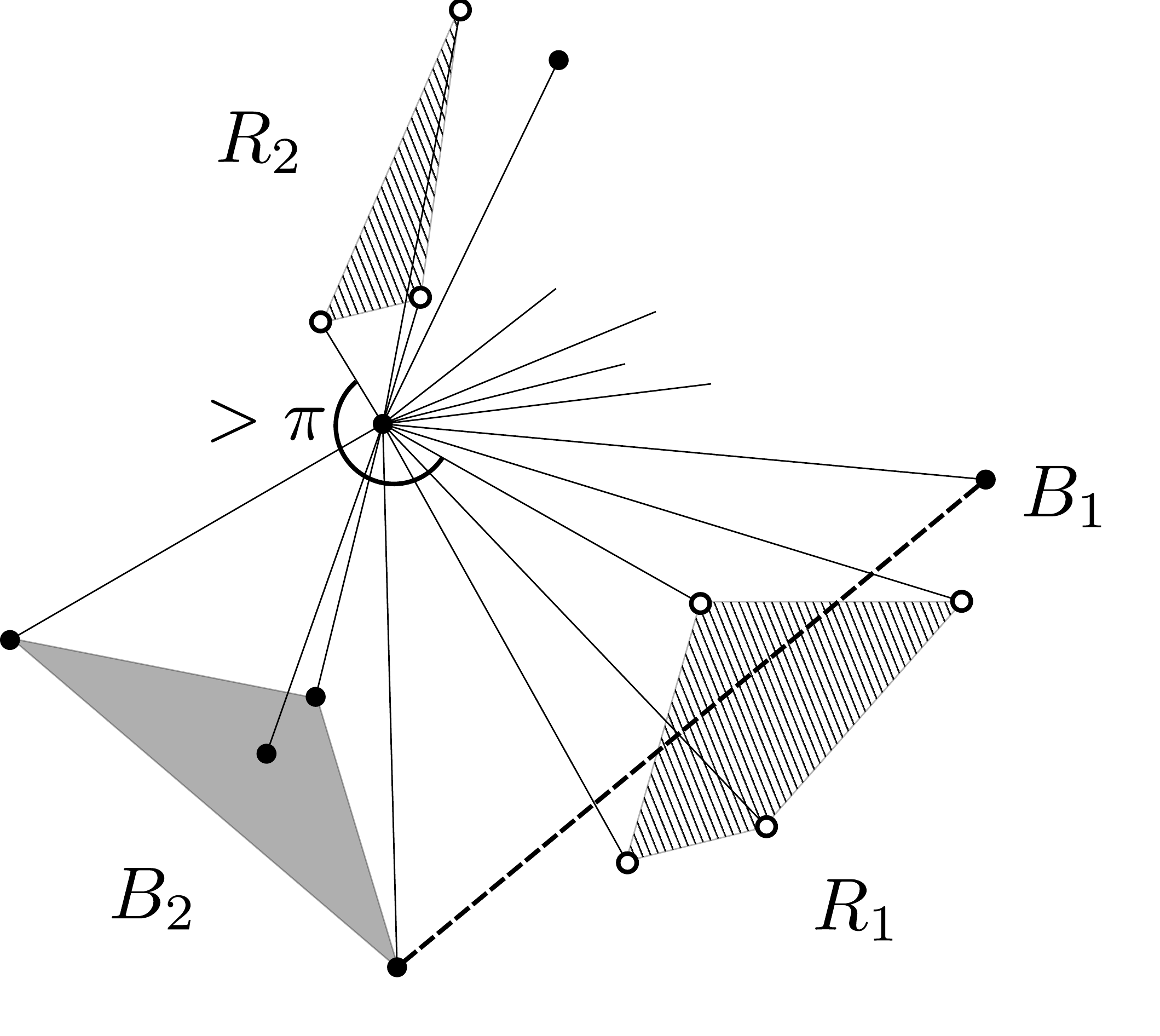}
\caption{Red monster-jump between $R_1$ and $R_2$.}
\label{fig:red-monster-jump}
\end{subfigure}
\quad
\begin{subfigure}[b]{0.4\textwidth}
\centering
\includegraphics[width=\textwidth]{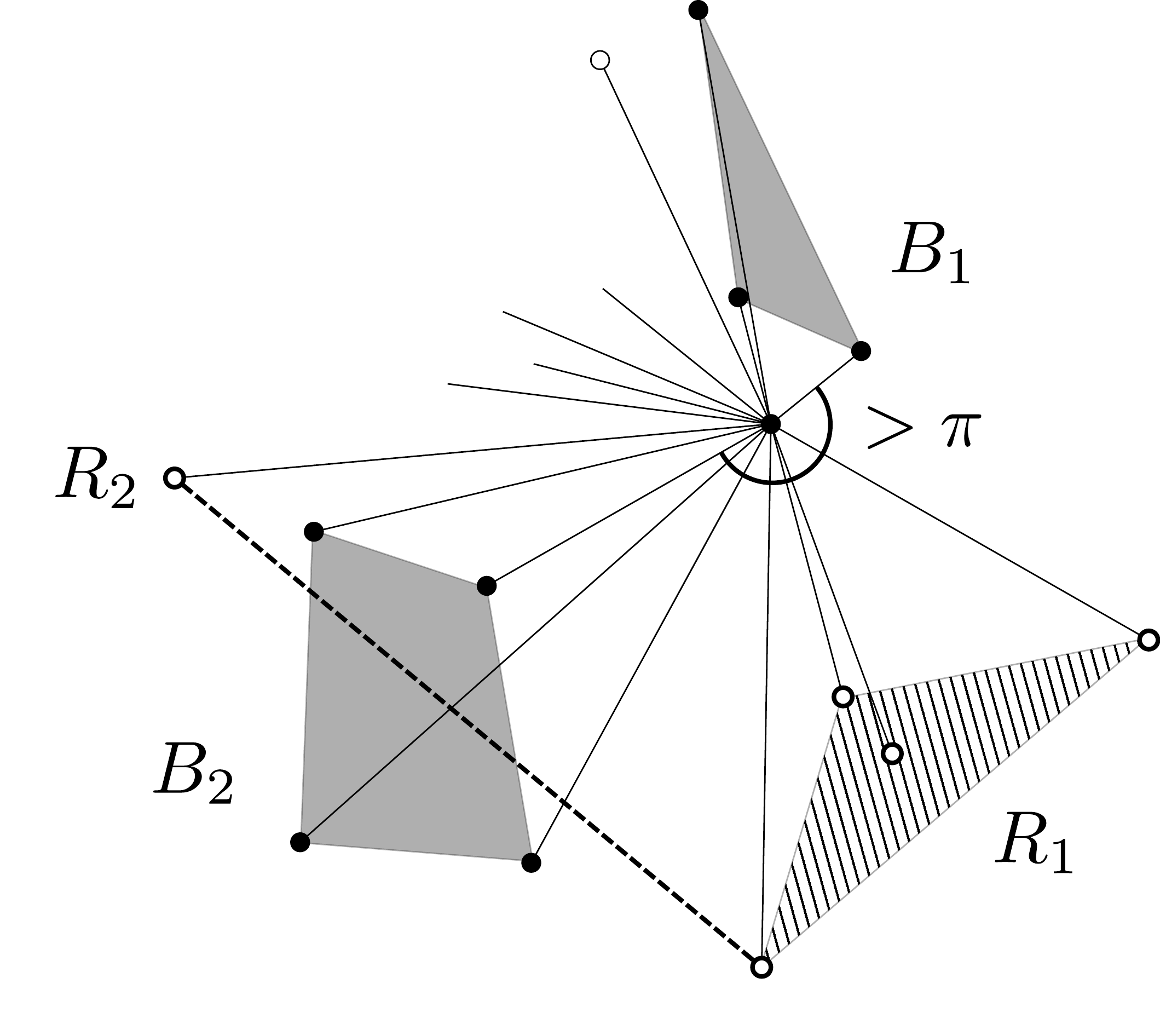}
\caption{Blue monster-jump between $B_1$ and $B_2$.}
\label{fig:blue-monster-jump}
\end{subfigure}
\caption{Red and blue monster-jumps}
\label{fig:monster-jumps}
\end{figure}

In this section we show that if $p$ is chosen such that the radial ordering about $p$ has no monster-jumps, then we can construct a jump configuration which can be completed to a pair of monochromatic spanning cycles with no 4-crossings. 

\subsection{Monster-Jumps and 4-forcings}
\begin{lemma}
  \label{lem:jump-config-exists}
  Suppose that, for each blob $X$, the angle from the last point in $X$ to the first point in the next blob of the same color is less than $\pi$. Then the collection of jump-edges consisting of, for each blob $X$, the edge from the last point of $X$ to the first point of the next blob of the same color, is a valid jump configuration.
  \begin{proof}
    For each blob $X$, let $a_X$ be the first point in $X$, and $b_X$ be the last point in $X$.
    Then for each blob $X$, if $X'$ is the next blob of the same color, then the angle from $b_X$ to $a_{X'}$ is less than $\pi$.
    Therefore every point on the line segment $(b_X,a_{X'})$ would, if added to the radial order, lie in the interval between $b_X$ and $a_{X'}$.
    If $Y$ is the blob before $X$ and $Y'$ is the blob after $X$, then $X$ is contained in the interval between the last point of $Y$ and the first point of $Y'$, which has angle less than $\pi$, so, for any point $z$ in the convex hull of $X$, if $z$ was added to the radial order, $z$ would lie between the first point of $X$ and the last point of $X$.
    Hence, $(b_X,a_{X'})$ can only intersect the convex hull of $X$ and $b_X$.
    Similarly, $(b_X,a_{X'})$ can only intersect the convex hull of $X'$ at $a_{X'}$.
    Therefore $(b_X,a_{X'})$ is a valid jump edge.

    If $X''$ is the next blob of the same color after $X'$, then the interval between $b_X$ and $a_{X'}$ and the interval between $b_{X'}$ and $a_{X''}$ are either disjoint, if $|X'|>1$, or meet at the point $a_{X'}=b_{X'}$, if $|X'|=1$.
    Therefore the line segments $(b_{X},a_{X'})$ and $(b_{X'},a_{X''})$ are either disjoint, if $|X'|>1$; or meet at $a_{X'}=b_{X'}$, if $|X'|=1$.
    Thus, the collection $J$ of jump edges consisting of, for each blob $X$ with next blob of the same color $X'$, $(b_X,a_{X'})$, is a valid jump configuration.
  \end{proof}
\end{lemma}

\begin{figure}
\centering
\includegraphics[width=9cm]{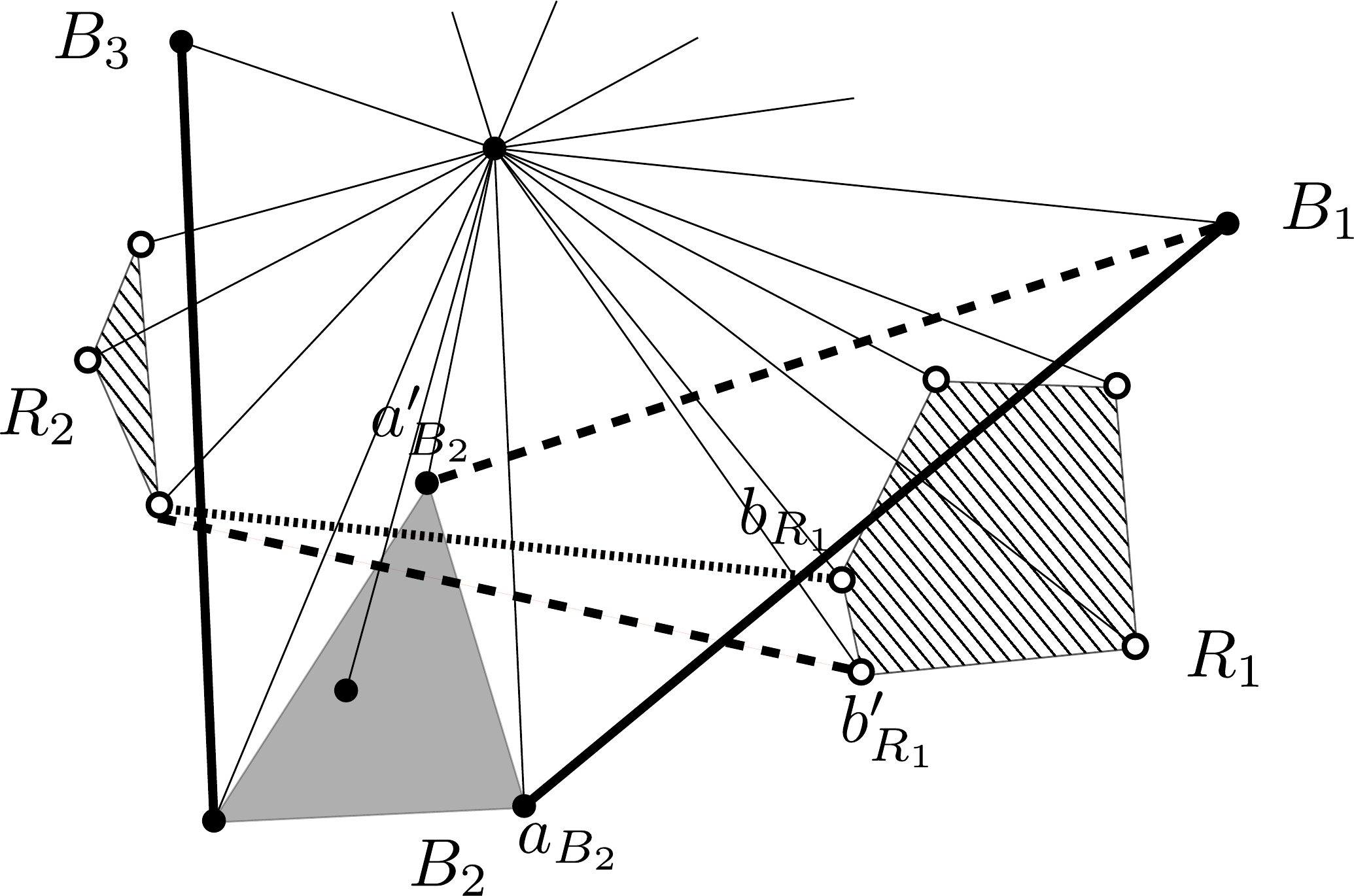}
\caption{In the proof of Lemma \ref{lem:no-4-forcing}, if there is a 4-forcing on the red jump edge between $R_1$ and $R_2$, then the blue-red crossing between the jump edge from $B_1$ to $B_2$ and the jump edge from $R_1$ to $R_2$ can be removed by replacing $b_{R_1}$ by $b'_{R_1}$ and replacing $a_{B_2}$ by $a'_{B_2}$.}
\label{fig:no-4-forcing-proof}
\end{figure}

\begin{lemma}
  \label{lem:no-4-forcing}
  If a radial order contains no red monster-jump and no blue monster-jump and, for each blob $X$, the angle from the last point in $X$ to the first point in the next blob of the same color is less than $\pi$, then there exists a jump configuration which contains no 4-forcing.
  \begin{proof}
    A \emph{blue-red crossing} in a jump configuration occurs when there are consecutive blobs $B_1,R_1,B_2,R_2$ such that $B_1$ and $B_2$ are blue and $R_1$ and $R_2$ are red, and the jump edge from $B_1$ to $B_2$ crosses the jump edge from $R_1$ to $R_2$.

    For each red blob $R$, let $a_R$ be the first point in $R$.
    For each blue blob $B$, let $b_B$ be the last point in $B$.

    For each red blob $R$ choose $b_R\in R$, and for each blue blob $B$ choose $a_B\in B$ such that the collection $J$ of edges consisting of, for each blob $X$, $(b_X,a_{X'})$, where $X'$ is the next blob of the same color, is a valid jump configuration, and that the number of blue-red crossings is minimized, with ties broken by minimizing the number of 4-crossings.
    Note that such a jump configuration exists by Lemma \ref{lem:jump-config-exists}.

    Suppose that the jump configuration $J$ contains a 4-forcing.
    By swapping the colors and reflecting the plane, if necessary (so that we preserve the fact that there are no red or blue monster-jumps), we may assume that there are consecutive blobs $B_1,R_1,B_2,R_2,B_3$ such that $B_1,B_2,B_3$ are blue and $R_1,R_2$ are red, and that the red jump edge $(b_{R_1},a_{R_2})$ from $R_1$ to $R_2$ is the center edge of a 4-forcing.
    Then $(b_{B_1},a_{B_2})$ and $(b_{B_2},a_{B_3})$ both cross $(b_{R_1},a_{R_2})$, $a_{B_2}$ and $b_{B_2}$ are on the same side of the line $\ell$ through $b_{R_1}$ and $a_{R_2}$, and $B_2$ contains points on both sides of $\ell$.
    
    If $b'_{R_1}$ is the last point in $R_1$, then replacing $(b_{R_1},a_{R_2})$ with $(b'_{R_1},a_{R_2})$ in $J$ will give another valid jump configuration $J'$, and cannot increase the number of blue-red crossings, because the only blue-red crossing that a jump edge from $R_1$ to $R_2$ could be involved in is with the jump edge from $B_1$ to $B_2$, which $(b_{R_1},a_{R_2})$ already crosses.
    Therefore, by minimality of $J$, $(b'_{R_1},a_{R_2})$ must be the center edge of a 4-forcing in $J'$.

    This means that $a_{B_2}$ and $b_{B_2}$ are both on the same side of the line $\ell'$ through $b'_{R_1}$ and $a_{R_2}$, and $B_2$ contains points on both sides of $\ell'$, and that $b_{B_1}$ and $a_{B_3}$ are on the opposite side of $\ell'$ from $a_{B_2}$ and $b_{B_2}$.
    There is some point $a'_{B_2}\in B_2$ on the same side of $\ell'$ as $b_{B_1}$, which $b_{B_1}$ sees with respect to $B_2$.
    The angle from $b_{B_1}$ to $a'_{B_2}$ is at most the angle from $b_{B_1}$ to the second to last point in $B_2$, which is less than $\pi$ because $B_1$ to $B_2$ is not a blue monster-jump.
    Thus $a'_{B_2}$ sees $b_{B_1}$ with respect to $B_1$, so $(b_{B_1},a'_{B_2})$ is a valid jump edge, and replacing $(b_{B_1},a_{B_2})$ by $(b_{B_1},a'_{B_2})$ in $J'$ gives a valid jump configuration $J''$.
    The only blue-red crossing which $(b_{B_1},a'_{B_2})$ can be involved in is with a jump edge from $R_1$ to $R_2$, and $(b_{B_1},a'_{B_2})$ does not cross $(b'_{R_1},a_{R_2})$, so $J''$ has fewer blue-red crossings than the chosen jump configuration $J$, a contradiction.
    See Figure \ref{fig:no-4-forcing-proof}.
  \end{proof}
\end{lemma}

\subsection{Avoiding Monster-Jumps}
We have shown that if we choose $p$ such that the radial order about $p$ gives us a jump configuration with no monster-jumps, then we can complete this jump configuration to a pair of monochromatic spanning cycles with no 4-crossings. In this section we show that it is possible to choose a point $p$ that avoids monster-jumps.
This is broken down into two cases: when the blue and red convex hulls properly overlap (Lemma \ref{lem:no-monster-jump-overlap}) and when the red convex hull contains the blue convex hull (Lemma \ref{lem:no-monster-jump-contains}).
The only other alternative is that the red and blue convex hulls are disjoint, in which case the desired spanning cycles exist trivially.

\begin{figure}
\centering
\includegraphics[width=4cm]{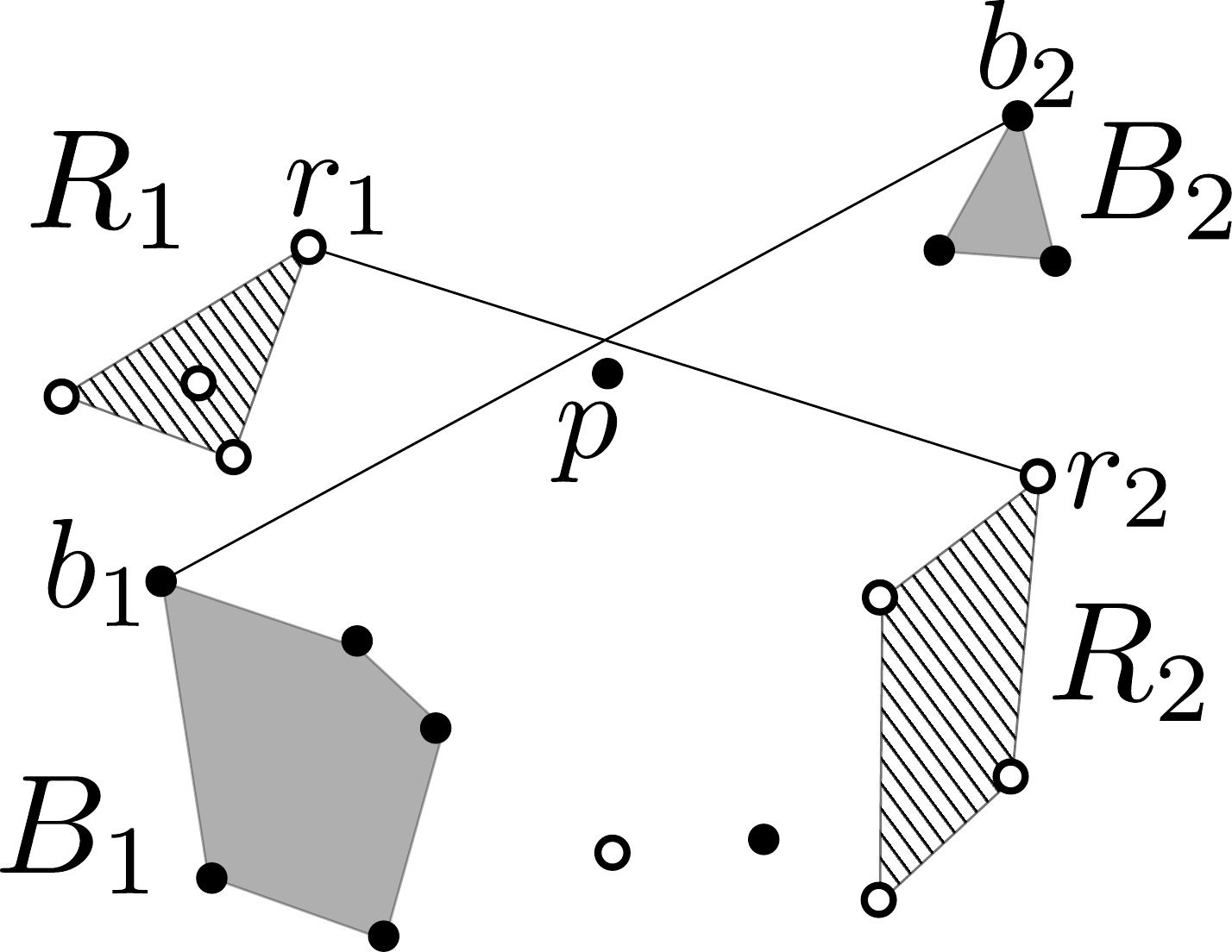}
\caption{The choice of $p$ when the red and blue convex hulls properly overlap (Lemma \ref{lem:no-monster-jump-overlap}).}
\label{fig:overlap-case}
\end{figure}

\begin{lemma}
  \label{lem:no-monster-jump-overlap}
  Suppose the bichromatic point set $S$ contains at least three red points and at least three blue points and that the convex hulls of the red points and of the blue points properly overlap; i.e., the intersection of the convex hulls is non-empty, and the blue convex hull is not contained in the red convex hull, nor vice-versa.
  Then after possibly swapping the color classes, there exists a point $p$ in the intersection of the interior of the convex hulls such that the radial order about $p$ contains neither a red nor a blue monster-jump.
  \begin{proof}
    Note that because the convex hulls properly overlap, the boundaries of the convex hulls must intersect at some point $q$.
    Let $(r_1,r_2)$ be the red segment containing $q$, such that $r_2$ follows $r_1$ in the clockwise ordering of the vertices of the red convex hull.
    Similarly, let $(b_1,b_2)$ be the blue segment containing $q$, such that $b_2$ follows $b_1$ in the clockwise ordering of the vertices of the blue convex hull.
    By swapping colors if necessary, we may assume that $b_1,r_1,b_2,r_2$ appear in that order in the (clockwise) radial order about $q$.
    Therefore $r_1$ is outside the blue convex hull and $b_2$ is outside the red convex hull.

    For all $\epsilon>0$, there is a point $p$ in the intersection of the interiors of the red and blue convex hulls such that $|p-q|<\epsilon$.

    We will show that for $\epsilon$ sufficiently small, there is no red or blue monster-jump in the radial order about $p$.
    If $\epsilon$ is sufficiently small, then the radial orders of the bichromatic point set with respect to $p$ and with respect to $q$ coincide.
    
    Any point between $b_1$ and $b_2$ in the radial order with respect to $p$ is outside the blue convex hull, and hence must be red.
    Therefore there is exactly one red blob $R_1$ between $b_1$ and $b_2$, and $r_1\in R_1$.
    Similarly there is exactly one blue blob, $B_2$, between $r_1$ and $r_2$, and $b_2\in B_2$.

    Let $B_1$ be the blue blob containing $b_1$ and $R_2$ be the red blob containing $r_2$.
    Note that $b_1$ is the last point in $B_1$ and $b_2$ is the first point in $B_2$, and the blob $R_1$ lies entirely on one side of the line through $b_1$ and $b_2$, so $R_1$ to $R_2$ is not a red monster-jump.
    Similarly, $r_1$ is the last point in $R_1$, $r_2$ is the first point in $R_2$, and $B_2$ lies entirely on one side of the line between $r_1$ and $r_2$, so $B_1$ to $B_2$ is not a blue monster-jump.
    
    Let $B$ be any blue blob which is not $B_1$, and let $B'$ be the next blue blob after $B$.
    If $|B'|=1$, then $B$ to $B'$ is not a blue monster-jump, so suppose $|B'|>1$.
    Then the last point, $b$, in $B$ and the second to last point, $b'$, in $B'$ lie in the interval between $b_2$ and $b_1$ in the radial order with respect to $p$.
    Note that $b\neq b_1$ because $B\neq B_1$ and $b'\neq b_1$, because $b_1$ is the not the second to last point in $B_1$.
    Therefore if $b'_1$ is the point after $b_1$ in the radial order with respect to $p$, then the angle from $b$ to $b'$ is at most the angle from $b'_1$ to $b_2$, which, for $\epsilon$ sufficiently small, is less than $\pi$.
    So, $B$ to $B'$ is not a blue monster-jump.
    By a symmetric argument, if $R$ is a red blob which is not $R_1$ and $R'$ is the next red blob after $R$ and $\epsilon$ is sufficiently small, then $R$ to $R'$ is not a red monster-jump.
    Therefore, if $\epsilon$ is sufficiently small, then the radial order of the bichromatic point set with respect to $p$ contains no red or blue monster-jump.
  \end{proof}
\end{lemma}

\begin{lemma}
  \label{lem:no-monster-jump-contains}
  Suppose the bichromatic point set $S$ contains at least three red points and at least three blue points and that the red convex hull contains the blue convex hull.
  Then there exists a point $p$ in the interior of the blue convex hull such that the radial order of the bichromatic point set with respect to $p$ contains no red or blue monster-jump.

  \begin{proof}
    Let $(b_1,\ldots,b_k)$ be the vertices of the blue convex hull in clockwise order.
    For $i\in\{1,\ldots,k\}$, let $H_i$ be the open half-plane bounded by the line through $b_i$ and $b_{i+1\pmod{k}}$ which contains no blue points.
    Note that $\bigcup_{i=1}^k H_i\setminus H_{i+1\pmod{k}}$ is the complement of the blue convex hull, and so, for some $i\in\{1,\ldots,k\}$, $H_i\setminus H_{i+1\pmod{k}}$ contains a red point.
    Without loss of generality, there is some red point $r_1\in H_k \setminus H_1$.
    Because $b_1$ is in the interior of the red convex hull, there is some red point in $H_1$; choose a red point $r_2$ in $H_1$ minimizing the angle from $r_2$ to $b_2$ with respect to $b_1$. 

\begin{figure}
\centering
\includegraphics[width=4.5cm]{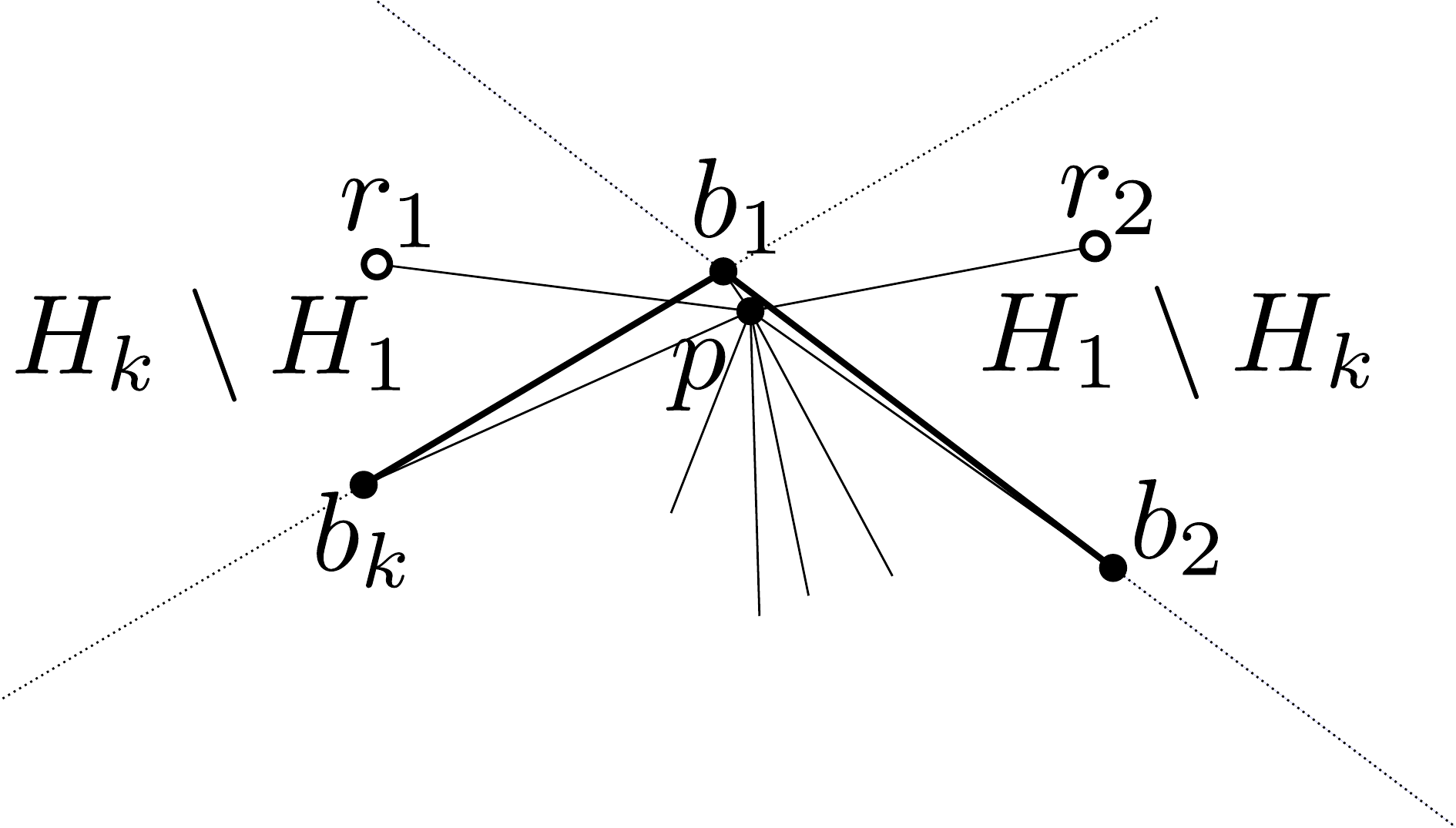}
\caption{Case 1: The blue convex hull is contained in the red convex hull, and
there are red points $r_1$ in $H_k\setminus H_1$ and $r_2$ in $H_1\setminus H_k$.}
\label{fig:containment-case-1}
\end{figure}

    \textbf{Case 1:} $r_2\in H_1\setminus H_k$ (see Figure \ref{fig:containment-case-1}). For $\delta>0$, let $v_\delta$ be the vector $b_k-b_1$ rotated counter-clockwise by $\delta$, and, for $\epsilon>0$, let $p=b_1 + \epsilon v_\delta$, and consider the radial order of the bichromatic points with respect to $p$.
    For $i\in\{1,2,k\}$, let $B_i$ be the blue blob containing $b_i$.
    For $i\in\{1,2\}$, let $R_i$ be the red blob containing $r_i$.

    If $\epsilon$ and $\delta$ are sufficiently small, $R_1$ lies entirely on one side of the line through $b_1$ and $b_k$, which are the first point of $B_1$ and last point of $B_k$ respectively, so $R_1$ to $R_2$ is not a red monster-jump.
    Again, if $\epsilon$ and $\delta$ are sufficiently small, $R_2$ lies entirely on one side of the line through $b_1$ and $b_2$, so $R_2$ to the next red blob is not a red monster-jump.
    If $R$ is a red blob that is not $R_1$ or $R_2$ and $R'$ is the next red blob and $\epsilon$ and $\delta$ are sufficiently small, then $R$ lies entirely on the same side of the line through $b_1$ and $b_2$ as $r_1$, so the angle between the second point of $R$ and the first point of $R'$ is at most the angle between the second point of $R$ and $r_1$, which is less than $\pi$.
    Therefore the radial order with respect to $p$ contains no red monster-jump.

    For $\epsilon$ and $\delta$ sufficiently small, the points before and after $b_1$ are both red, so $B_1=\{b_1\}$.
    Therefore $B_k$ to $B_1$ is not a blue monster-jump.
    Let $b_k'$ be the previous blue point before $b_k$ in the radial order.
    For any blue blob $B$ that is not $B_k$, if $B'$ is the next blue blob, then the angle between the last point $B$ and the second last point of $B'$ is at most the angle between $b_1$ and $b'_k$, which is less than $\pi$ if $\epsilon$ and $\delta$ are sufficiently small.
    Therefore the radial order with respect to $p$ contains no blue monster-jump.

\begin{figure}
\centering
\includegraphics[width=5cm]{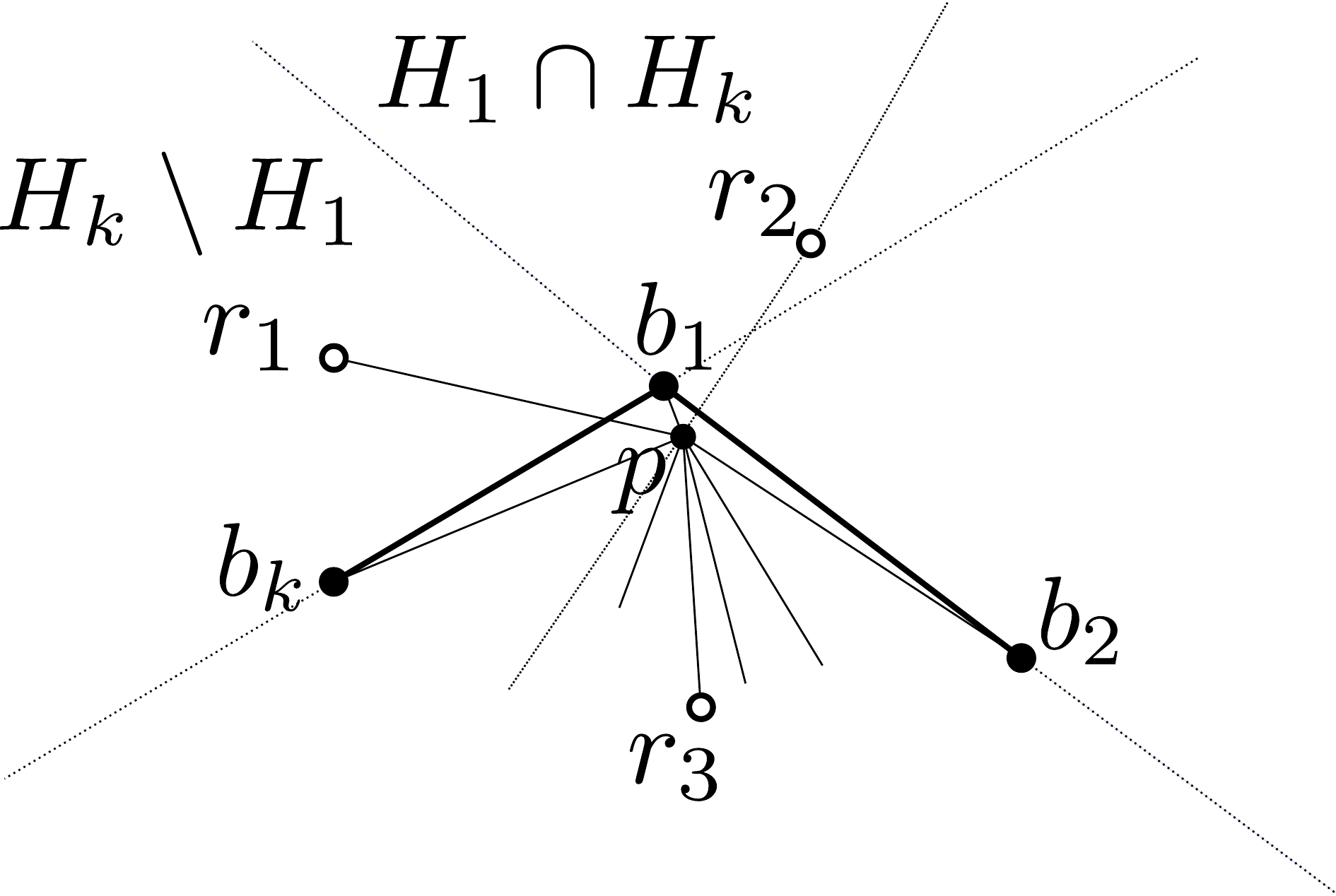}
\caption{Case 2: The blue convex hull is contained in the red convex hull and
there are red points $r_1$ in $H_k\setminus H_1$ and $r_2$ in $H_1\cap H_k$.}
\label{fig:containment-case-2}
\end{figure}

    \textbf{Case 2:} $r_2\in H_1\cap H_k$ (see Figure \ref{fig:containment-case-2}).
    Let $\ell$ be the line through $r_2$ and $b_1$.
    Because $b_1$ is in the interior of the red convex hull, there is a red point $r_3$ on the opposite side of $\ell$ from $r_1$.
    However, $r_3$ cannot lie between $r_2$ and $b_2$ in the radial order with respect to $b_1$, so $r_3$ lies between $b_2$ and $b_k$ in the radial order with respect to $b_1$, and the angle from $r_2$ to $r_3$ with respect to $b_1$ is less than $\pi$.
    
    For $\delta>0$, let $v_\delta$ be the vector from $r_1$ to $b_1$ rotated counter-clockwise by $\delta$.
    As in case 1, for $\epsilon>0$, let $p=b_1+\epsilon v_\delta$, and consider the radial order of the bichromatic points with respect to $p$.
    For $i\in\{1,2,k\}$, let $B_i$ be the blue blob containing $b_i$, and, for $i\in\{1,2\}$, let $R_i$ be the red blob containing $r_i$.

    If $\epsilon$ and $\delta$ are sufficiently small, then $R_1$ lies entirely on one side of the line between $b_k$ and $b_1$, which are the last point of $B_k$ and first point of $B_1$ respectively, so $R_1$ to the next red blob is not a red monster-jump.
    For $\epsilon$ and $\delta$ sufficiently small, $b_1$ is the point before $r_2$ and $b_2$ is the point after $r_2$ so $R_2=\{r_2\}$. Therefore, $R_2$ to the next red blob is not a red monster-jump.
    If $R$ is a red blob that is neither $R_1$ nor $R_2$, and $R'$ is the next red blob, then $R$ lies entirely on the same side of the line through $b_1$ and $b_2$ as $r_1$, so the angle from the second point in $R$ to the first point in $R'$ is at most the angle from $b_2$ to $r_1$, which is less than $\pi$.
    Therefore the radial order with respect to $p$ contains no blue monster-jump.

    If $\epsilon$ and $\delta$ are sufficiently small, then the points before and after $b_1$ are both red, so $B_1=\{b_1\}$.
    Therefore $B_k$ to $B_1$ is not a blue monster-jump.
    If $\epsilon$ and $\delta$ are sufficiently small, then $r_3$ lies between $b_2$ and $b_k$, and the angle from $b_1$ to $r_3$ is less than $\pi$.
    The blob $B_2$ ends before $r_3$, so the angle from $b_1$ to any point in $B_2$ is less than $\pi$, and $B_1$ to $B_2$ is not a blue monster-jump.
    If $B$ is a blue blob that is not $B_k$ or $B_1$, and $B'$ is the next blue blob, then $B$ and $B'$ lie between $b_2$ and $b_k$ in the radial order, so the angle from any point in $B$ to any point in $B'$ is less than $\pi$, and hence $B$ to $B'$ is not a blue monster-jump.
    Therefore the radial order with respect to $p$ contains no blue monster-jump.

    In both cases, $p$ can be chosen such that the radial order with respect to $p$ contains no red monster-jump and no blue monster-jump.
  \end{proof}
\end{lemma}

\begin{proof}[Proof of Theorem \ref{thm:main}]
If the red convex hull and the blue convex hull are disjoint, then any pair of red and blue spanning cycle will be disjoint, so assume the red and blue convex hulls intersect.
By Lemma \ref{lem:complete-jump-config}, it suffices to show that for some point $p$, the radial order about $p$ contains a jump-configuration with no 4-forcing.
By Lemma \ref{lem:no-4-forcing}, it suffices to show that there exists a point $p$ such that the radial order about $p$ contains no red or blue monster-jump.

Either the red and blue convex hulls properly overlap, or the blue convex hull is contained in the red convex hull, or the red convex hull is contained in the blue convex hull.
If the red and blue convex hulls properly overlap, then by Lemma \ref{lem:no-monster-jump-overlap}, there is a point $p$ such that the radial order with respect to $p$ contains no red or blue monster-jump.
If the red convex hull is contained in the blue convex hull, we may swap the colors, so that the blue convex hull is contained in the red convex hull, and in that case, by Lemma \ref{lem:no-monster-jump-contains}, there is a point $p$ such that the radial order with respect to $p$ contains no red or blue monster-jump.
\end{proof}

\bibliography{references}{}
\bibliographystyle{plain}
\end{document}